\documentclass[11pt]{article}
\usepackage{amsmath,amssymb,amsthm,verbatim}

\title{A Monotone Function Given by a Low-Depth Decision Tree that is not an Approximate Junta}
\author{Daniel M. Kane}

\newcommand{\floor}[1]{\left\lfloor #1 \right\rfloor}

\newcommand{\pr}{\textrm{Pr}}

\newcommand{\E}{\mathbb{E}}

\newcommand{\poly}{\textrm{poly}}

\newcommand{\s}{\mathcal{S}}
\newcommand{\expref}[2]{#1 \ref{#2}}

\newtheorem{theorem}{Theorem}
\newtheorem{proposition}[theorem]{Proposition}

\newtheorem{lemma}[theorem]{Lemma}
\newtheorem{conjecture}[theorem]{Conjecture}

\newtheorem*{observation}{Observation}

\begin{document}
\maketitle

\section{Introduction}

In \cite{learningmonotone}, O'Donnell and Servedio show that any monotone function given by a depth-$d$ decision tree can be learned to constant accuracy from random samples in $\poly(n,2^d)$ time.  The impact of this result is somewhat lessened by an apparent lack of interesting monotone functions given by low-depth decision trees.  In particular, it was independently suggested by Elad Verbin and by Rocco Servedio and Li-Yang Tan, that all such functions might be approximated by functions on few variables (see \cite{Simons}, page 10).

\begin{conjecture}\label{TVConj}
For every $\epsilon>0$ and every monotone function $f:\{0,1\}^n\rightarrow \{0,1\}$ given by a depth-$d$ decision tree, there is a $k$-junta, $g$, for $k=\poly_\epsilon(d)$ so that $f$ and $g$ agree on all but an $\epsilon$-fraction of inputs.
\end{conjecture}

In this note, we disprove the above conjecture, and in particular provide an example of a monotone low-degree function that is not well approximated by any small junta.  In particular we prove:

\begin{theorem}\label{MainThm}
There exists a constant $\epsilon>0$ so that for every positive integer $d$, there exists a $k=\exp(\Omega(\sqrt{d}))$ and a monotone function $f:\{0,1\}^n\rightarrow\{0,1\}$ given by a depth-$d$ decision tree, so that for every $k$-junta $g$, $f$ and $g$ disagree on at least an $\epsilon$-fraction of inputs.
\end{theorem}

In fact it is known that the bound on $k$ in \expref{Theorem}{MainThm} is tight up to the constant in the exponent.  In particular, it is shown in \cite{learningmonotone} that any monotone function given by a depth-$d$ decision tree has total influence $I(f) = O(\sqrt{d})$.  We combine this with the main result of \cite{juntaapprox}, which says that any boolean function $f$ can be $\epsilon$-approximated by a $k$-junta for $k=\exp(O(I(f)/\epsilon)).$  Combining these results we find that:

\begin{observation}
If $f$ is a monotone function given by a depth-$d$ decision tree, and if $\epsilon>0$, then there is a $k$-junta $g$ that agrees with $f$ on all but an $\epsilon$ fraction of inputs for $k=\exp(O(\sqrt{d}/\epsilon))$.
\end{observation}

The function we construct to show \expref{Theorem}{MainThm} will combine ideas from two previous constructions, the monotone addressing function and Talagrand's function.

The monotone addressing function is defined by
$$
f(x_1,\ldots,x_{d-1},y_0,\ldots,y_{2^{d-1}-1}) = \begin{cases} 1 &\textrm{ if } \sum x_i > \floor{(d-1)/2}
\\ y_{x_0\ldots x_{d-1}}&\textrm{ if } \sum x_i = \floor{(d-1)/2} \\ 0 &\textrm{ if } \sum x_i < \floor{(d-1)/2}
\end{cases}.
$$
This is an example of a monotone function given by a depth-$d$ decision tree that depends on exponentially many variables, and thus provides us with a good starting point. The monotone addressing function fails to provide a counter-example to \expref{Conjecture}{TVConj} though since it agrees with the majority function except on a set of measure $O(1/\sqrt{d})$.

Given the bound on the total sensitivity of a low-depth monotone function, we know that any $f$ satisfying the conditions of \expref{Theorem}{MainThm} must not only have near the maximum possible total influence for a low-depth monotone function, but also must not be approximable by a function with much lower total influence.  Because of this restriction, our construction will look somewhat similar to a construction of Talagrand in \cite{similarexample}.  In particular, Talagrand constructs a monotone function $f$ on $\{0,1\}^d$ so that on a constant fraction of inputs, $f$ has sensitivity (i.e. the number of coordinates such that changing the input at that coordinate would change the output of $f$) $\Omega(\sqrt{d})$.  Since, as is easily seen, the average sensitivity over all inputs is equal to the total influence, this is as large as possible.  On the other hand, this condition tells us that $f$ retains large average sensitivity even after ignoring any $\epsilon$-fraction of inputs for sufficiently small constant $\epsilon$.  Talagrand's function fails to provide a counter-example to \expref{Conjecture}{TVConj} on its own, because it is already a $d$-junta.

\section{The Construction}

In order to define the function $f$ with the properties specified by \expref{Theorem}{MainThm}, we first introduce some background notation.  We let $d,t$ and $m$ be integers with $t=\Theta(\sqrt{d})$ and $m=\Theta(2^t)$.  We furthermore assume that $2^{-t}m$ is sufficiently small given the value of $t /\sqrt{d}$.  We let $\s=(S_1,\ldots,S_m)$ be a random sequence of sets, where the $S_i$ are chosen independently and uniformly from the set of subsets of $\{1,2,\ldots,d-1\}$ of size exactly $t$.  Given this $\s$, we define the function $T_\s$ on $\{0,1\}^{d-1}$ as follows:
$$
T_{\s}(x_1,\ldots,x_{d-1}) = \{ 1\leq i \leq m: x_j = 1 \textrm{ for all } j \in S_i\}.
$$
We will hereafter abbreviate $T$ by suppressing the explicit dependence on $\s$, and abbreviate $(x_1,\ldots,x_{d-1})$ by $x$.

We finally define $f$ as
$$
f_{\s}(x_1,\ldots,x_{d-1},y_1,\ldots,y_m) = \begin{cases} 1 & \textrm{if } |T(x)|\geq 2 \\ 0 & \textrm{if } |T(x)|=0 \\ y_i & \textrm{if } T(x)=\{i\} \end{cases}.
$$
Again, we will often suppress the dependence of $f$ on $\s$.  It is clear that $f$ is monotone.  Furthermore, $f$ is given by a depth-$d$ decision tree, since after fixing the values of the $x_i$, the value of $f$ depends on at most one more coordinate.  In the next Section, we show that $f$ cannot be approximated by any $k$-junta for small $k$.

Note that Talagrand's function is given (for appropriately chosen $\s$) by
$$
G(x_1,\ldots,x_{d-1}) = \begin{cases} 1 & \textrm{if } |T(x)|\geq 1 \\ 0 & \textrm{if } |T(x)|=0 \end{cases}.
$$

\section{Approximation Bounds}

\expref{Theorem}{MainThm} will follow from the following Proposition:

\begin{proposition}\label{mainProp}
There exists an $\epsilon>0$ so that for $f_{\s}$ defined as above, with constant probability over the choice of $\s$, $f$ is not $\epsilon$-approximated by any $k$-junta for $k=o(2^t)$.
\end{proposition}

A key step in our proof will be to show that with constant probability $f$ actually depends on one of the $y_i$.

\begin{lemma}\label{TSizeLem}
With $T$ as above,
$$
\pr_{\s,x}(|T_{\s}(x)| = 1) = \Omega(1).
$$
\end{lemma}
\begin{proof}
We will show the further claim that
\begin{equation}\label{TBoundEqn}
\E\left[|T_{\s}(x)|(2-|T_{\s}(x)|) \right] = \Omega(1).
\end{equation}
Since the term in the expectation is positive only if $|T|=1$, this will complete our proof.  We note that
\begin{align*}
\E\left[|T_{\s}(x)|\right] &= \sum_{i=1}^m \pr(i \in T_{\s}(x))\\
& = \sum_{i=1}^m \pr(x_j=1 \textrm{ for all } j\in S_i) \\
& = m 2^{-t}.
\end{align*}
On the other hand, we have that
\begin{align*}
\E\left[|T_{\s}(x)|(|T_{\s}(x)|-1) \right] & = \sum_{i\neq j} \pr(i,j \in T_{\s}(x))\\
& = \sum_{i\neq j} \pr(i\in T_{\s}(x)) \pr(j\in T_{\s}(x)| i \in T_{\s}(x))\\
& = \sum_{i\neq j} 2^{-t} \pr(x_\ell =1 \textrm{ for all } \ell\in S_j | x_\ell =1 \textrm{ for all } \ell\in S_i ).
\end{align*}
To compute this conditional probability we let $S_j=\{a_1,\ldots,a_t\}$ where the $a_i$ are picked randomly from $\{1,2,\ldots,d-1\}$ without replacement.  We compute it as the product
$$
\prod_{k=1}^t \pr(x_{a_k}=1 | x_{a_1}=\ldots=x_{a_{k-1}}=1 \textrm{ and }x_\ell =1 \textrm{ for all } \ell\in S_i ).
$$
These probabilities are approximated by first fixing the values of $S_i$ and $a_1,\ldots,a_{k-1}$.  After additionally fixing the value of $a_k$, the probability in question becomes $1$ if $a_k\in S_i$ and $1/2$ otherwise.  Thus the probability that $x_{a_r}=1$ is
$$
(1+\pr(a_r\in S_i))/2 = \left( 1 + \frac{|S_i \backslash \{a_1,\ldots,a_{r-1}\}|}{d-r}\right)/2 = 1/2 + O(t/d).
$$
Hence the probability that $j\in T_{\s}(x)$ given that $i\in T_{\s}(x)$ is
$$
(1/2+O(t/d))^t = 2^{-t} \exp(O(t^2/d)).
$$
Therefore, we have that
$$
\E\left[|T_{\s}(x)|(|T_{\s}(x)|-1) \right] = \sum_{i\neq j} 2^{-2t}\exp(O(t^2/d)) \leq (2^{-t}m)^2\exp(O(t^2/d)).
$$
Therefore, we have that
\begin{align*}
\E\left[|T_{\s}(x)|(2-|T_{\s}(x)|) \right] &= \E\left[|T_\s(x)|\right]-\E\left[|T_\s(x)|(|T_\s(x)|-1)\right]\\
&\geq (2^{-t}m) - (2^{-t}m)^2\exp(O(t^2/d))\\
&= (2^{-t}m)\left(1-(2^{-t}m)\exp(O(t^2/d))\right).
\end{align*}
As long as $2^{-t}m$ is bounded below by a constant and above by $\exp(-O(t^2/d))/2$, this is $\Omega(1)$.
\end{proof}

We are now ready to prove \expref{Proposition}{mainProp}.  By \expref{Lemma}{TSizeLem}, we note that with constant probability over $\s$, that $\pr_x(|T(x)|=1) = \Omega(1)$.  For such $\s$, we claim that $f$ has the desired property.  In particular we claim the following:

\begin{lemma}\label{finalApproxLem}
If $f$ is as above and $g$ is a $k$-junta, then
$$
\pr(f(x,y)\neq g(x,y)) \geq \frac{\pr_x(|T(x)|=1) - k 2^{-t}}{2}.
$$
\end{lemma}
\begin{proof}
This follows from the simple observation that if, after fixing the value of $x$, we have that $T=\{i\}$ where $g$ does not depend on $y_i$, then $\pr_y(f(x,y)\neq g(x,y)) = 1/2$.  This is because after further conditioning on the values of all $y_j$ for $j\neq i$, $g$ becomes a constant function (by assumption) and $f$ takes the values $0$ and $1$ each with probability $1/2$.  Therefore we have that
\begin{align*}
\pr & (f(x,y)\neq g(x,y))\\ & \geq \frac{\pr(T(x)=\{i\}\textrm{ for some } i, \textrm{ and } g \textrm{ does not depend on }y_i)}{2} \\
& = \frac{\pr(|T(x)|=1) - \pr(T(x)=\{i\}\textrm{ for some }i, \textrm{ and } g \textrm{ depends on }y_i)}{2}\\
& = \frac{\pr(|T(x)|=1) - \sum_{i:g\textrm{ depends on }y_i} \pr(T(x) = \{i\})}{2}\\
& \geq \frac{\pr(|T(x)|=1) - \sum_{i:g\textrm{ depends on }y_i} \pr(i \in T(x)) }{2}\\
& = \frac{\pr(|T(x)|=1) - \sum_{i:g\textrm{ depends on }y_i} 2^{-t} }{2}\\
& \geq \frac{\pr_x(|T(x)|=1) - k 2^{-t}}{2}.
\end{align*}
\end{proof}

Proposition \ref{mainProp} and Theorem \ref{MainThm} now follow immediately.

\section*{Acknowledgements}

I would like to thank Ryan O'Donnell for making me aware of this problem, and for his help with finding appropriate references for this paper.  This research was done with the support of an NSF postdoctoral fellowship.


\begin{thebibliography}{[99]}

\bibitem{juntaapprox} E. Friedgut \emph{Boolean Functions with Low Average Sensitivity Depend on Few coordinates}, Combinatorica Vol. 18(1), pp. 27-36, 1998.

\bibitem{Simons} R. O'Donnell \emph{Open Problems in Analysis of Boolean Functions} http://arxiv.org/abs/1204.6447.

\bibitem{learningmonotone} R. O'Donnell, R. Servedio \emph{Learning monotone decision trees in polynomial time}, SIAM Journal on Computing Vol. 37(3), pp. 827-844, 2007.

\bibitem{similarexample} M. Talagrand \emph{How much are increasing sets positively correlated?}, Combinatorica
Vol. 16(2), pp. 243-258, 1996.



\end{thebibliography}
\end{document}